\documentclass[11pt,letterpaper]{article}
\usepackage[margin=1in]{geometry}
\usepackage{amsmath,amssymb,amsfonts,setspace,url,mathrsfs}
\usepackage{bm,enumerate}
\usepackage[boxed,linesnumbered,noend]{algorithm2e}
\DontPrintSemicolon
\let\oldnl\nl
\newcommand{\nonl}{\renewcommand{\nl}{\let\nl\oldnl}}
\usepackage{algpseudocode}
\usepackage{latexsym}
\usepackage{amsthm}
\usepackage{subfig}
\usepackage{float}
\usepackage{stmaryrd}
\usepackage{fancybox}
\usepackage{bigstrut,array,multirow}
\usepackage{here}
\usepackage{color}
\usepackage{url}
\usepackage{hyperref}
\usepackage[capitalise]{cleveref}
\usepackage{tikz}
\usepackage{comment}

\newcommand{\jsay}[1]{{\color{red} Josh says: #1}}

\newcommand{\GpI}{\textsc{GpI}\xspace}
\newcommand{\pblmDef}{\textsc{Group Epimorphism Problem}\xspace}
\newcommand{\GrS}{\textsc{GpEpi}\xspace}
\newcommand{\Epi}{\text{epivalent}\xspace}
\newcommand{\RepEpi}{\textsc{SInducedRepEpivalent}\xspace}
\newcommand{\codeIso}{\textsc{Code Isomorphism}\xspace}
\newcommand{\ep}{\text{exponent}\xspace}

\newcommand{\Z}{\mathbb{Z}}

\newcommand{\aut}{\mathrm{Aut}}

\newcommand{\gen}[1]{\langle #1 \rangle}

\newcommand{\mybar}[1]{\lambda}

\newcommand{\ord}[1]{\textrm{ord}{#1}}

\newcommand{\F}{\field}

\newcommand{\poly}{\mathrm{poly}}

\newcommand{\rk}{\mathrm{rank}}

\newcommand{\NP}{\mathrm{NP}}

\newcommand{\GL}{\mathrm{GL}}
\newcommand{\field}{\mathbb{F}}

\DeclareMathOperator{\Soc}{Soc}
\newtheorem{theorem}{Theorem}
\newtheorem{corollary}{Corollary}

\newtheorem{problem}{Problem}
\newtheorem{observation}{Observation}

\newtheorem{claim}{Claim}
\newtheorem{lemma}{Lemma}
\newtheorem{open}{Open Question}
\theoremstyle{definition}
\newtheorem{definition}{Definition}

\newcommand{\Hall}[2]{\mathrm{Coprime}(#1,#2)}
\newcommand{\ElemAb}{\mathrm{ElemAb}}
\newcommand{\Ab}{\mathrm{Ab}}
\newcommand{\Abp}{\mathrm{Ab_p}}
\newcommand{\Cyc}{\mathrm{Cyc}}

\newenvironment{proof-sketch}{\trivlist\item[]\emph{Brief proof sketch}:}%
{\unskip\nobreak\hskip 1em plus 1fil\nobreak$\Box$
\parfillskip=0pt%
\endtrivlist}
\begin{document}

\sloppy
\title{Algorithms for Finite Group Epimorphism Testing\footnote{A preliminary version of this article appeared in the proceedings of the 53rd EATCS International Colloquium on Automata, Languages, and Programming (ICALP 2026)}}
\author{
Joshua A. Grochow\\
Department of Computer Science, University of Colorado Boulder \\ Department of Mathematics, University of Colorado Boulder\\ jgrochow@colorado.edu 
\and Pranjal Srivastava\\
Department of Math and Computing \\ Indian Institute of Information Technology Vadodara\\
pranjal.srivastava194@gmail.com
\and Dhara Thakkar\\
Graduate School of Mathematics\\
Nagoya University\\thakkar\_dhara@math.nagoya-u.ac.jp}
\date{}
\maketitle

\begin{abstract}
The $\pblmDef$ $(\GrS)$ asks, given two finite groups $G_1$ and $G_2$, whether there exists a surjective group homomorphism, or \emph{epimorphism}, from $G_1$ to $G_2$. When the input groups are given by their multiplication (Cayley) tables, the problem admits a quasipolynomial-time algorithm in general, but little is known about its complexity for structured classes of finite groups. In this paper, we study the computational complexity of $\GrS$ for several well-studied classes of finite groups.

Our main results are polynomial-time epimorphism tests for several classes of groups for which polynomial-time isomorphism testing was previously known:
\begin{itemize}
\item[$\bullet$] Groups with Abelian normal Hall subgroups with cyclic complement

\item[$\bullet$] Groups with (product of) elementary Abelian normal Hall subgroup with elementary Abelian complement.

\item[$\bullet$] Groups with some constraints on their Abelian chief factors.
\end{itemize}
\end{abstract}

\section{Introduction}
The study of computational problems in finite groups has a long history in both (computational) algebra and theoretical computer science. Among the most well-known is the Group Isomorphism problem ($\GpI$), which asks whether two given groups are isomorphic. 
From the perspective of worst-case guarantees,\footnote{In this paper we focus on algorithms with worst-case guarantees as a function of the order of the group, though we occasionally comment on heuristics that will likely improve performance in practice. Because of our focus, we can without loss of generality assume that the multiplication (Cayley) tables are given as input. We refer the reader to the Handbook of Computational Group Theory \cite{Holt+2005}, as well as, e.\,g., the works of Holt, Eick, O'Brien, Leedham-Green, and Hulpke 
\cite{EHHSurvey,BEick+2002,Eick+2002,BEick+1999,CH+2003,HulpkePerfect} for research into practical algorithms in more succinct input models.} this problem has been extensively studied for various classes of groups, including Abelian groups, groups with normal Hall subgroups, and groups without Abelian normal subgroups \cite{GrIBCQ,GrIBCGQ,GrIGenus2,GrISylowTower,GrIExtensionandCohomology,GrICountingToDecision,GItoGrI-GQ19,LeGall_classical_iso,GrIWilson,QST,GrIrosenbaumBidirectional,GrISun,GrIFOCS24}. Efficient algorithms in these settings often rely on exploiting the underlying algebraic structure, such as decompositions into simpler subgroups or the use of representation theory.

A related problem is the $\pblmDef$ ($\GrS$), which, given two finite groups $G_1$ and $G_2$, asks to decide whether there exists an \emph{epimorphism} (=surjective homomorphism)\footnote{In some concrete categories, such as the category of rings, the concepts of epimorphism and surjective homomorphism differ, but in the category of groups they are the same \cite{Linderholm}, so we do not dwell on the two definitions.} from $G_1$ to $G_2$. When $G_1$ and $G_2$ are given by their Cayley tables, one can easily design an $n^{\log_p n+O(1)}$-time algorithm for the $\GrS$ problem, where $|G_1|=n$ and $p$ is the smallest prime divisor of $n$. Unlike isomorphisms, epimorphisms allow the domain group to be larger and potentially more complex than the codomain, introducing new algorithmic and structural challenges.

The $\GrS$ problem has been studied previously, primarily in the setting of \emph{finitely presented  groups}, where the focus is often on decidability and hardness questions. In this context, Remeslennikov showed that the epimorphism problem is undecidable when both the source and target groups are non-Abelian nilpotent groups given by finite presentations \cite{Epimo5}. Subsequent work has identified classes of target groups for which epimorphism testing is decidable or computationally hard. Friedl and L\"oh proved decidability results when the target group is virtually cyclic or a direct product of an Abelian and a finite group \cite{Epimo3}, while Elder, Shen, and Weiss showed that the problem is $\NP$-complete for similar classes of targets \cite{EpiArmin}. Kuperberg and Samperton further established $\NP$-hardness when the target is a fixed finite non-Abelian simple group \cite{Epimo4}.

In contrast, comparatively little is known about the computational complexity of $\GrS$ for \emph{finite groups given by their Cayley tables}. While $\GpI$ and related problem have been extensively studied in this model, efficient algorithms for epimorphism testing remain largely unexplored. Understanding which classes of finite groups admit polynomial-time  algorithms for $\GrS$ is a natural and significant open direction in computational group theory and complexity theory. 

Motivated by this, in this paper, we investigate the computational complexity of $\GrS$ for several well-studied classes of groups. We focus on classes for which structural properties can be exploited algorithmically, and for which efficient isomorphism testing is known. In particular, we study $\GrS$ for groups with Abelian normal Hall subgroups and groups without Abelian composition factors. Our results build on and extend techniques from the literature on $\GpI$, while also introducing a new representation-theoretic problem and developing algorithmic frameworks specific to epimorphism testing.

Before coming to our contributions we mention a few points on the general complexity of \GrS, that also serve to motivate some of the classes of groups we consider. If $|G|=|H|$, then epimorphisms $G \to H$ and isomorphisms are the same thing, so $\GpI \leq_m^p \GrS$. In particular, this means that for any class of groups we cannot put \GrS for that class into $\mathsf{P}$ without doing the same for \GpI, hence we focus on classes of groups for which \GpI is already known to be in $\mathsf{P}$. It is clear that the problem is in $\mathsf{NP}$ (the epimorphism serves as a witness), and the generator-enumerator technique gives an $|G|^{\log |G| + O(1)}$-time algorithm for \GrS. Unlike isomorphism problems, we do not know whether \GrS is in $\mathsf{coAM}$ (though the quasi-polynomial-time upper bound nonetheless implies that it is not $\mathsf{NP}$-complete unless the Exponential Time Hypothesis is false). Just as \GpI reduces to \textsc{Graph Isomorphism}, \GrS reduces to \textsc{Graph Epimorphism}, but the latter is $\mathsf{NP}$-complete, even for target graphs of size $O(1)$ (epimorphisms to the triangle graph are the same as proper 3-colorings that use all three colors) \cite{graphsurjectionvikas1997, graphsurjection_GFMS}. It is a somewhat interesting question whether \GrS is in $\mathsf{P}$ when the target group has size $O(1)$.

\subsection{Our Contributions}
We study $\GrS$ for several well-studied classes of finite groups. For each such class, we design a deterministic polynomial-time algorithm for $\GrS$. We now summarize our main results.

To put our results in context, we recall some of the significant progress on algorithms with worst-case guarantees for \GpI, and use this opportunity to introduce some notation we will use throughout the paper. We begin with some notation: for two group classes ${\rm X}$ and ${\rm Y}$, $\Hall{\rm X}{\rm Y}$ is the class of groups with a normal subgroup ${N}$ from ${\rm X}$ and a complement $H$ from ${\rm Y}$, such that $|N|$ and $|H|$ are coprime. 
Let $\Ab$ be the class of Abelian groups and let $\Abp$ be the class of Abelian $p$-groups, where $p$ is a prime. Let $\ElemAb$ be the class of elementary Abelian groups, and let $\Cyc$ be the class of cyclic groups. 

Over the last two decades, two of the most general classes of groups for which polynomial-time algorithms for \GpI have been developed fall into the following lines of work, and we make progress on solving \GrS along these lines:
\begin{itemize}
\item \textbf{Coprime extensions.} Le Gall \cite{LeGall_classical_iso} solves \GpI for groups in $\Hall{\Ab}{\Cyc}$. Qiao, Sarma, and Tang \cite{QST} extended this to $\Hall{\Ab}{\ElemAb}$. This was extended to groups with Abelian Sylow towers by Babai \& Qiao \cite{BabaiQiao}, and to tame Sylow towers in Grochow \& Qiao \cite{GQtame}. We refer the reader to those papers for definitions as we do not yet handle \GrS for such groups here.

We do solve \GrS for groups in $\Hall{\Ab}{\Cyc}$ (cf. Le Gall \emph{ibid.}) and $\Hall{\prod \ElemAb}{ \ElemAb}$, and discuss challenges to extending our results to $\Hall{\Ab}{\ElemAb}$ (cf. Qiao--Sarma--Tang, \emph{ibid.}). 

\item \textbf{Fitting-free groups.} Isomorphism of groups without Abelian normal subgroups, a.k.a. semisimple or Fitting-free, was solved in \cite{GrIBCQ} (building on \cite{GrIBCGQ}). This was extended to more general classes of groups in \cite{GQcoho}.

Although we do not yet solve \GrS for arbitrary Fitting-free groups, we handle two classes of groups that are more restrictive, yet headed in this direction. Namely, we solve \GrS (1) for groups with no Abelian composition factors, 
and (2) for pairs $(G,H)$ where $H$ is Fitting-free and the Abelian chief factors of $G$ have limited multiplicity. 
We discuss some challenges in extending this to all Fitting-free groups in Sec.~\ref{sec:conclusion}.

\end{itemize}

For all the aforementioned results except the last, the classes of groups involved are closed under homomorphic images, and membership in the class is recognizable in polynomial time. If $\mathcal{C}$ is any class of finite groups with these properties, then solving \GrS for pairs $(G,H)$ with $G,H \in \mathcal{C}$ is equivalent to solving \GrS for pairs $(G,H)$ where $G \in \mathcal{C}$ and $H$ is arbitrary.



\subparagraph*{Efficient epimorphism testing of $\Hall{\Ab}{\Cyc}$.} We begin by considering the class $\Hall{\Ab}{\Cyc}$, consisting of groups that are extensions of an Abelian group by a cyclic group. This class serves as a first step beyond Abelian groups for $\GrS$. Polynomial-time algorithms for $\GpI$ are known for this class due to Le Gall \cite{LeGall_classical_iso}. \Cref{thm-AbelianCyclic} shows that $\GrS$ can be solved efficiently for groups in $\Hall{\Ab}{\Cyc}$. 
Although we leverage some of the same algebraic structure used by Le Gall, we also use additional results from representation theory and combinatorial techniques (maximum matching in bipartite graphs).


\begin{theorem}\label{thm-AbelianCyclic}
There is a polynomial-time algorithm for $\GrS$, restricted to pairs of groups $(G_1,G_2)$ such that $G_1$ is in $\Hall{\Ab}{\Cyc}$. If there is an epimorphism from $G_1$ to $G_2$, our algorithms computes one such. 
\end{theorem}

\subparagraph*{Efficient epimorphism testing of $\Hall{\prod \ElemAb}{\ElemAb}$.}  Next, we study $\GrS$, for groups in the class $\Hall{\prod \ElemAb}{\ElemAb}$, which consists of groups with a normal Hall subgroup that is a direct product of elementary Abelian groups, complemented by another elementary Abelian group. To develop an algorithm for this class, we first consider the subclass, $\Hall{\ElemAb}{\ElemAb}$.

To solve $\GrS$ for groups in $\Hall{\ElemAb}{\ElemAb}$, we reduce it to a problem in representation theory of finite groups. We now describe the specific formulation of this problem that captures the $\GrS$.

A \emph{representation} of a group $G$ is a homomorphism from $G$ to a general linear group. When the group $G$ is a finite elementary Abelian normal subgroup of a larger group, the condition of epimorphism testing of large groups with respect to their component groups naturally translates into questions about representations (\Cref{thm-Taunt-Surjection}). This leads to an interesting algorithmic problem in representation theory that is of independent interest. We define this problem below after introducing the following definition. 

Let $V$ and $W$ be two finite-dimensional vector spaces over a field $\mathbb{F}$ of dimension $d$ and $d'$, respectively. Let $\tau:G\to \GL(d, \field)$ and $\gamma:G \to \GL(d',\field)$ be two representations (see \Cref{sec-prelims}). We say that $\tau$ and $\gamma$ are \emph{$\Epi$} if there exists a surjective linear transformation $\rho:V \to W$ such that $\rho \tau(g) = \gamma(g) \rho$ for every $g \in G$.     

\begin{problem}($\RepEpi$)\label{prob-Epi}
Given two representations, $\tau: H_1 \to \GL(m_1, \Z_p)=\aut(\Z_p^{m_1})$ and $\gamma: H_2 \to \GL(m_2, \Z_p)=\aut(\Z_p^{m_2})$, such that $H_i$ are elementary Abelian $q$-groups, and $p,q$ are distinct primes, determine whether there exists a surjective homomorphism $\phi: H_1 \to H_2$ such that the induced representation $\tau$ and $\gamma \circ \phi$ are epivalent, i.\,e., there exists a surjective linear map $\rho: \Z_p^{m_1} \to \Z_p^{m_2}$ such that $\rho \circ\tau(h)=\gamma(\phi(h))\circ \rho$, $\forall h \in H_1$.
\end{problem}

The following theorem shows that $\RepEpi$ captures the essential difficulty of epimorphism testing for groups in $\Hall{\ElemAb}{\ElemAb}$.

\begin{theorem}\label{thm-EEtoSInduced}
For groups $G_1, G_2 \in \Hall{\ElemAb}{\ElemAb}$, $\GrS$ is polynomial-time many-one equivalent to $\RepEpi$.
\end{theorem}

To solve $\RepEpi$ in the above setting, we further reduce it to the well-known \textsc{Code Isomorphism} problem (\Cref{prob-codeIso}). 
This is similar to the technique used to solve \GpI for this class of groups in polynomial time \cite{QST}, but in our case reducing to \textsc{Code Isomorphism} requires additional work, since we start with an epimorphism problem and are reducing to an isomorphism problem (whereas in Qiao--Sarma--Tang \cite{QST} they were reducing from one isomorphism problem to another).

\begin{theorem}\label{thm-RepEpiTOCode}
Using the notation of \Cref{prob-Epi}, when $H_i$ are elementary Abelian of rank $d$, and the representations have $m = \max\{m_1,m_2\}$, the problem $\RepEpi$ reduces to \textsc{CodeIso} for $d$-dimensional codes in $\F^m$.
\end{theorem}

The above theorem enables us to derive the following result.

\begin{corollary} \label{cor-ElemElem}
There is a polynomial-time algorithm for $\GrS$, restricted to pairs of groups $(G_1,G_2)$ such that $G_1 \in \Hall{ \ElemAb}{\ElemAb}$. 

\end{corollary}

We then extend this result to obtain an efficient algorithm for testing epimorphism between groups in $\Hall{\prod \ElemAb}{\ElemAb}$.

\begin{theorem}\label{thm-ProdElem}
There is a polynomial-time algorithm for $\GrS$ restricted to pairs of groups $(G_1,G_2)$ such that $G_1 \in \Hall{\prod \ElemAb}{\ElemAb}$.
\end{theorem}


\subparagraph*{Efficient epimorphism testing of groups with limited Abelian composition factors.}

We next study the $\GrS$, for groups with no Abelian composition factors, that is, finite groups whose composition factors are all non‑Abelian simple groups. This class includes two important classes of groups, namely the \emph{direct product of non-Abelian simple groups}, and \emph{almost simple groups}. It forms a well-structured subclass of \emph{groups with no nontrivial Abelian normal subgroups} (a.k.a. Fitting-free groups). Groups with no nontrivial Abelian normal subgroups have been widely studied in various algorithmic contexts, including efficient isomorphism testing (see e.g., \cite{CH+2003,GrIBCGQ,GrIBCQ}).

Groups with no Abelian composition factors has also been studied extensively in group theory, particularly in the context of structural classification, enumeration, and understanding their minimal faithful permutation degree (see, e.g., \cite{wtAbCF-Benjamin, wtAbCF-Bercov, wtAbCF-Luca, wtAbCF-Robert, wtAbCF-Lucchini}). While algorithmic problems such as $\GpI$ and the Minimum Generating Set problem have been explored for various non‑solvable group classes including Fitting-free groups, the computational complexity of $\GrS$ even for groups with no Abelian composition factors has not been addressed previously. Our main result for this class is the following.

\begin{theorem}\label{thm-wtAbCF}
There is a polynomial-time algorithm for $\GrS$, restricted to pairs of groups $(G_1,G_2)$ such that $G_1$ has no Abelian composition factors.
\end{theorem}

We can extend this slightly in the direction of Fitting-free groups that allow Abelian composition factors of limited ``multiplicity;'' see Cor.~\ref{cor:Abwidth} and the definitions preceding it.

\subparagraph*{Organization.} 
The remainder of the paper is organized as follows. We begin with preliminaries in \Cref{sec-prelims}. In \Cref{sec:condition for EpiTesting}, we develop a criterion for epimorphism testing for large groups, reducing the problem to normal Hall subgroups and their complements. In \Cref{sec-abCyclic}, we prove \Cref{thm-AbelianCyclic}. In \Cref{sec-EE}, we establish \Cref{thm-EEtoSInduced}, \Cref{thm-RepEpiTOCode}, and \Cref{thm-ProdElem}. Finally, in \Cref{sec-wtAbCF}, we prove \Cref{thm-wtAbCF}. In \Cref{sec:conclusion}, we conclude with several open problems.

\section{Preliminaries}\label{sec-prelims}
In this section, we recall some relevant terminology required for the paper. An interested reader may refer to standard references such as \cite{Rotman1995,serre1977} for more details.

\subparagraph*{Group theory.} We consider groups with finitely many elements. The \emph{order} of group $G$ is the number of elements in a group $G$, denoted by $|G|$. The order of an element $g \in G$, denoted as $\ord(g)$, is the smallest positive integer $m$ such that $g^m=1$. The \emph{exponent} of $G$ is the smallest positive integer $m$ such that $g^m=1$ for all $g \in G$. Let $H$ be a subgroup of $G$, if for all $a \in G$, $a^{-1}Ha=H$ then we say that $H$ is a normal subgroup of $G$ (denoted by $H \triangleleft G$). A \emph{normal Hall} subgroup $N$ is a normal subgroup of $G$ with $\gcd(|N|,|G/N|)=1$. For any subgroup $H$ and any normal subgroup $K$ of $G$ we denote the subgroup $\{hk \,\vert  h\in H,k\in K\}=\{kh \,\vert  h\in H,k\in K\}$ by $HK$.  

Let $N$ be a normal subgroup of $G$. We say that $G$ is \emph{semidirect product} of $N$ by a subgroup $H$ of $G$, denoted as $G= N \rtimes H$ if $G=NH$ and $N \cap H=1$. Since $N \triangleleft G$, this gives rise to a homomorphism $\tau:H \to \aut(N)$ by $h \to \chi_h$, where $\chi_h: N \to N$ by $n \to h^{-1}nh$. In this case, we write $G=N \rtimes_\tau H$. We say $N$ is a \emph{normal subgroup} of this decomposition and $H$ the \emph{complement} of $N$ in $G$. Conversely, given two groups $N$ and $H$ and a homomorphism $\tau:H\rightarrow \aut(N)$, we can define a group $G:=\{(n,h) \mid n \in N, h \in H\}$ with the product operation, $(n_1,h_1)(n_2,h_2)=(n_1\chi_{h_1}(n_2),h_1h_2)$, for all $(n_1,h_1),(n_2,h_2)\in G$. This gives a construction of the outer semidirect product $G=N \rtimes_{\tau} H$. For a normal Hall subgroup $N$ of $G$, the Schur--Zassenhaus theorem guarentees that a complement subgroup $H \leq G$ such that $\gcd(|H|, |N|) = 1$ and $G = N \rtimes_{\tau} H$ exists. Furthermore, if $H$ and $K$ are complements of $N$, then $H$ and $K$ are conjugate (see, e.g., \cite{Rotman1995}).




\subparagraph*{Socle series and chief series.}
A \emph{minimal normal subgroup of $G$} is a normal subgroup $N \unlhd G$ with no nontrivial subgroups that are also normal in $G$; equivalently, if $M \unlhd G$ and $M \leq N$, then $M \in \{1,N\}$. A standard exercise is that every minimal normal subgroup of a finite group is a direct power $T^n$ of some finite simple group $T$. The \emph{socle} of a group $G$, denoted $\Soc(G)$ is the (characteristic) subgroup generated by all minimal normal subgroups of $G$. The \emph{socle series} of $G$ is defined by $\Soc^1(G) := \Soc(G)$ and $\Soc^i(G)$ is the unique subgroup of $G$ such that $\Soc^i(G) / \Soc^{i-1}(G) = \Soc(G / \Soc^{i-1}(G))$. 

A \emph{normal series} in $G$ is a series $1 = N_0\leq N_1 \leq N_2 \leq \dotsb \leq N_k = G$ such that $N_i \unlhd G$ for all $i$. A \emph{chief series} is a normal series that cannot be further refined (a so-called ``maximal'' normal series). Equivalently, a chief series is a normal series such that $N_i / N_{i-1}$ is a minimal normal subgroup of $G / N_{i-1}$ for all $i$. If $N_{*}$ is a chief series, then the factor groups $N_{i} / N_{i-1}$ are called \emph{chief factors} of $G$. Each chief factor is of the form $T^n$ for some simple group $T$ and some integer $n$, and the multiset of isomorphism types of chief factors is well-defined for $G$, independent of the chief series (see e.g., \cite[Theorem 8.4.4]{marshall-hall-theory}).

The socle series can be refined to a chief series by taking the minimal normal subgroups $M_1,\dotsc,M_k$ of $G / \Soc^i(G)$, taking their pre-images $\hat{M}_1,\dotsc,\hat{M}_k$ in $G$, and then inserting them in between $\Soc^i(G)$ and $\Soc^{i+1}(G)$ in any order, viz. $\Soc^i(G) \unlhd \Soc^i(G) M_1 \unlhd \Soc^i(G) M_1 M_2 \unlhd \dotsb \unlhd \Soc^i(G) M_1 M_2 \dotsb M_{k-1} \unlhd \Soc^{i+1}(G)$. We say $G$ has \emph{socle length} $\ell$ if $\ell$ is the least integer such that $G = \Soc^{\ell}(G)$.

\subparagraph*{The Code Isomorphism problem.}Let $\mathbb{F}$ be a field. A $d$-dimensional linear subspace $C \subseteq \mathbb{F}^n$ is called a \emph{linear code} of dimension $d$. A generating matrix of a linear code $C$, of dimension $d$, is a $d \times n$ matrix with row vectors being a basis of $C$. We will also use $C$ to denote the generating matrix of a linear code $C$. 

\begin{problem}($\codeIso$)\label{prob-codeIso}
Given two generating matrices $C_{d \times n}$ and $D_{d \times n}$ over the field $\mathbb{F}$, the $\codeIso$ Problem asks to determine the existence of $T \in {\rm GL}(d, \mathbb{F})$ and a permutation matrix $P_{n \times n}$ such that $T_{d \times d} C_{d \times n}P_{n \times n}=D_{d \times n}$.  
\end{problem}

\subparagraph*{Representation theory.} Let $V$ be a vector space of dimension $d$ over $\field$. A \emph{representation} of a group $G$ is a homomorphism $\tau:G \rightarrow \GL(d, \mathbb{F})$. A subspace $W$ of $V$ is called $G$-\emph{invariant} if for all $g \in G$ and $w \in W$, we have $\tau_g(w) \in W$. A representation $\tau:G \rightarrow {\GL}(d, \mathbb{F})$ is called \emph{irreducible} if the only $G$-invariant subspaces of $V$ are ${0}$ and $V$. 

Let $V$ and $W$ be two finite-dimensional vector spaces over a field $\mathbb{F}$ of dimension $d$ and $d'$, respectively. We say that two representations $\tau: G \rightarrow \GL(d, \mathbb{F})$ and $\gamma: G \rightarrow \GL(d', \mathbb{F})$ are \emph{equivalent} if there is an invertible linear map $f:V \rightarrow W$ such that  $f\circ \tau_g = \gamma_g \circ f$ for all $g \in G$. 

Let $N$ be an elementary abelian $p$-group and let $G= N \rtimes_\tau H$ be a group, where $\tau:H \to \aut(N)$ is an action of $H$ on $N$. Since $N$ is an elementary abelian $p$-group, $N \cong \mathbb{Z}_p^d$, and $\aut(N)=\GL(d, \mathbb{F}_p)$. Then the homomorphism $\tau$ can be seen as a representation of $H$ given by $\tau:H \to \aut(N)=\GL(d, \mathbb{F}_p)$.


\subparagraph*{Linear algebra and Modules.} We say that matrix $A_{d \times n}$ of rank $d$ is in standard form if $A=[I_d \vert A']$, i.e., the first $d$ columns of $A$ form the identity matrix. We can perform row operations on $A$ followed by a permutation of the columns to transform $A$ to a standard form, i.e., we can find $B$ and $B'$ such that $BAB'$ is in standard form.

Let $R$ be a ring  and $M,N$ be $R$-modules. A map $\phi:M \rightarrow N$ is an \emph{$R$-module homomorphism} if $\phi(\alpha x+y)=\alpha \phi(x)+\phi(y)$ for all $x,y \in M$ and $\alpha\in R$. A module $M$ is \emph{indecomposable} if $M = M_1 \oplus M_2$ implies $M_1 = 0$ or $M_2 = 0$.  

Let $G$ be a group. A (left) $G$-module is an abelian group $A$ on which $G$ acts by additive maps on the left, i.e., a left $G$-module consists of an abelian group $A$ together with a left group action $\tau: G\times A\rightarrow A$ such that $g \cdot (a_1+a_2)=g\cdot a_1+g\cdot a_2$, for all $a_1,a_2 \in A $ and $g \in G$. Similarly, we can define right $G$-module. Unless explicitly
mentioned otherwise the term ``$G$-module'' will always mean ``left $G$-module''. We refer reader to \cite{Weibel+1994,DF-Algebra} for more details on group modules.


\section{Condition for epimorphism testing} \label{sec:condition for EpiTesting}
In this section, prove that epimorphism testing for groups with normal Hall subgroups reduces to that for the normal Hall subgroups and their complements. Our proof is based on the ideas discovered by Taunt \cite[Theorem 3.3]{Taunt1955}. A similar result for isomorphism testing between groups was established using the same ideas (see e.g., \cite[Theorem 3]{QST}).

\begin{theorem}\label{thm-Taunt-Surjection}
Given $G_1=N_1 \rtimes_{\tau} H_1$, $G_2=N_2 \rtimes_{\gamma} H_2$, where $N_1, N_2$ are normal Hall subgroups. Then there is an epimorphism from $G_1$ to $G_2$ if and only if there exist an epimorphism  $\rho:N_1 \to N_2, $ and an epimorphism $\phi:H_1 \to H_2$, such that, $\forall h \in H_1$, 
\begin{equation} \label{equation}
 \rho \circ\tau(h)=\gamma(\phi(h))\circ \rho.
\end{equation}
\end{theorem}

\begin{proof}
    $(\implies)$ Let $f:G_1 \to G_2$ be an epimorphism from $G_1$ to $G_2$ such that $G_1/\ker f \cong G_2$. Let $L=\ker f \cap N_1$ and $K=\ker f \cap H_1$. It is easy to see that $\ker f = L\rtimes_{\tau} K$. 

We now show that the action of $K$ on $N_1/L$, given by $(nL)^{k}=knk^{-1}L$ for $k \in K$, $n\in N_1$, is trivial. Since $f(k)=1$ we have $f(knk^{-1}) = f(n)$, which implies that $n^{-1}knk^{-1} \in \ker f$. Clearly $n^{-1}knk^{-1} \in N_1$, therefore $n^{-1}knk^{-1} \in L$ and $knk^{-1}L=nL$. Define an action $\tau'$ of $H_1/K$ on $N_1/L$ given by $\tau'(hK):=\tau(h)L$. Since $K$ acts on $N_1/L$ trivially, $\tau'$ is well-defined. Indeed, if $hK=h'K$, then $h'=hk$ for some $k\in K$. $\tau'(h'K)(nL)
=\tau(h')(n)L=\tau(hk)(n) L=(\tau(h) \tau(k))(n)L=\tau(h)(\tau(k)(n)L)=\tau(h)(n)L=\tau'(hK)(nL)$. 

Consider a map $\pi:\frac{N_1 \rtimes_{\tau} H_1}{L \rtimes_{\tau} K} \to \frac{N_1}{L} \rtimes_{\tau'} \frac{H_1}{K}$ defined by $\pi(n_1,h_1)=(n_1 L, h_1 K)$. It is easy to see that $\pi$ is an isomorphism. Moreover, $\frac{N_1}{L} \rtimes_{\tau'} \frac{H_1}{K} \cong N_2 \rtimes_\gamma H_2$, with $|N_1/L|=N_2$ and $|H_1/K|=H_2$. By \cite[Theorem 3.3]{Taunt1955}, there exist an isomorphism $\overline{\rho}: N_1/L \to N_2 $ and an isomorphism $\overline{\phi}: H_1/K \to H_2 $ such that, $\forall hK \in H_1/K$, $\overline{\rho} \circ\tau'(hK)=\gamma(\overline{\phi}(hK))\circ \overline{\rho}$. By lifting $\overline{\phi}, \overline{\rho}, \tau'$ to $\phi, \rho$ and $\tau$ respectively, via the canonical epimorphism we obtain $ \rho \circ\tau(h)=\gamma(\phi(h))\circ \rho$.

$(\impliedby)$
Consider a map $f:G_1 \to G_2$ defined by $f(n,h)=(\rho(n), \phi(h))$. It is clear that $f$ is an epimorphism from $G_1$ to $G_2$. Now, consider 
\begin{align*}
f((n,h)(n',h'))&=f(n\tau_{h}(n'),h h')\\
&=(\rho(n\tau_{h}(n')), \phi(h h'))\\
&=(\rho(n)\rho(\tau_{h}(n'),\phi(h)\phi(h'))\\
&=(\rho(n)\gamma_{\phi(h)}(\rho(n')),\phi(h)\phi(h'))\\
&=(\rho(n),\phi(h))(\rho(n'),\phi(h'))\\
&=f(n,h)f(n',h').
\end{align*}
Hence, $f$ is an epimorphism from $G_1$ to $G_2$.
\end{proof}

\section{Efficient epimorphism testing of $\Hall{\Ab}{\Cyc}$}\label{sec-abCyclic}
In this section, we design an algorithm for finding an epimorphism between two groups $G_1$ and $G_2$ such that $G_1 \in \Hall{\Ab}{\Cyc}$. Before we present our algorithm, we prove certain preliminary results which will be required in the algorithm. 

Consider two groups $G_1, G_2 \in \Hall{\Ab}{\Cyc}$ defined as $G_1=N_1 \rtimes_{\tau} H_1$ and $G_2=N_2 \rtimes_{\gamma} H_2$, where $N_1$ and $N_2$ are Abelian, and $H_1$ and $H_2$ are cyclic. Moreover, $\gcd(|H_1|,|N_1|)=1$ and $\gcd(|H_2|,|N_2|)=1$. Suppose there exists an epimorphism $\rho: N_1 \to N_2$ that satisfies the condition of \Cref{equation} for some epimorphism $\phi: H_1 \to H_2$. Let $K=\ker(\phi)$ be the unique subgroup of $H_1$ whose order is $|H_1|/|H_2|$. Define a subgroup $L$ of $N_1$ generated by the set $\{v^{\tau(k)}-v \mid v\in N_1$, $k\in K\}$. It is clear that $L \leq \mathrm{ker}(\rho)$.

The following observation on $L$ is of independent interest, and we realized it in the course of coming up with the proofs below. Although it is ultimately not required for our proofs, as it may be of future use we include its proof in \Cref{app}.

\begin{observation}\label{lem:max-indecomposable}
With the notations above, $L$ is a direct sum of maximal indecomposable $H_1$-submodules of $N_1$.
\end{observation}

Let $N$ be an Abelian group, but not necessary elementary Abelian. A $(\Z/p^k \Z)[H]$-module is a $\Z H$-module $N$ where the exponent of $N$ (the LCM of the orders of the elements of $N$) divides $p^k$.

\begin{lemma}\cite[Corollary 1.2]{thevenaz1981} \label{lem:Thevenaz}
Let $H$ be a finite group. If $p$ is coprime to $|H|$, then any indecomposable $(\Z/p^k \Z)[H]$-module is generated (as an $H$-module) by a single element. 
\end{lemma}

\begin{lemma}\label{lem:compute-max-indecomposable}
Let $N$ be an Abelian group (not necessarily elementary Abelian), and $H$ a cyclic group acting on $N$ by automorphisms. Moreover, suppose $\gcd(|N|,|H|)=1$. Then a decomposition of $N$ into its indecomposable $H$-submodules can be computed in polynomial time in $|N|,|H|$.
\end{lemma}


\begin{proof}
For each $v \in N$, we compute the $H$-submodule $\langle v \rangle_H$ of $N$. For each such submodule, we then use the algorithm of Cioc\v{a}nea-Teodorescu \cite[Theorem 1.1]{Teodorescu} to test whether $\langle v \rangle_H$ appears as a direct summand in a direct sum decomposition of $N$ as an $H$-module, and if so, to find a complement, that is, another $H$-submodule $U \leq N$ such that $N = \langle v \rangle_H \oplus U$. As soon as one such direct sum decomposition is found, the algorithm is then called recursively on $\langle v \rangle_H$ and $U$ to see if they can be further decomposed. Furthermore, if $N$ is decomposable, then there must exist a $v \in N$ such that $\langle v \rangle_H$ is a direct summand, for $N$ must contain an indecomposable direct summand by induction on its size, and for a prime $p$ dividing $|N|$, by Lemma~\ref{lem:Thevenaz}, in this setting, indecomposable $H$-modules are generated by a single element. If no such direct sum decomposition is found, then $N$ is an indecomposable $H$-module, and the recursion stops.
\end{proof}

In practice, some pre-processing can be done beforehand; for example, at the start, $N$ can be split into a direct sum of its Sylow $p$-subgroups for different primes $p$. And within each Sylow $p$-subgroup, if $\langle v \rangle_H = p \langle u \rangle_H$ for some $u \in N$, then $\langle v \rangle_H$ is not a direct summand of $N$ and can be excluded from the above search.

We are now ready to prove \Cref{thm-AbelianCyclic}. 

\vspace{.1cm}

\noindent{\emph{\bf Proof of \Cref{thm-AbelianCyclic}.}} 
First we check that $G_2$ is in $\Hall{\Ab}{\Cyc}$ (otherwise reject). In order to test whether there is an epimorphism between $G_1$ and $G_2$, we first run the algorithm from \cite[Theorem 1]{QST} on the inputs $G_1$ and $G_2$ to obtain all normal Hall subgroups and their complements. We then select an arbitrary normal Hall subgroup $N_1$ of $G_1$ (respectively, $N_2$ of $G_2$) and its complement $H_1$ (respectively, $H_2$) such that $|N_1| \geq |N_2|$, $|H_1| \geq |H_2|$, $N_1$ and $N_2$ are Abelian, and $H_1$ and $H_2$ are cyclic. All these conditions can be verified in polynomial time for each normal Hall subgroup and all of its complements. If no such pair  of normal Hall subgroup and its complement exists, then by \Cref{thm-Taunt-Surjection}, we conclude that there is no epimorphism from $G_1$ to $G_2$. Otherwise, we proceed as follows.

Since every finite cyclic group has a unique subgroup of any given order dividing its order, there exists a unique subgroup $K \triangleleft H_1$ such that $|H_1/K|=|H_2|$. If no such $K$ exists, then the algorithm will reject. Moreover, $K$ can be computed in polynomial time. Let $\tau$ (respectively, $\gamma$) be a conjugation action of $H_1$ on $N_1$ (respectively, $H_2$ on $N_2$). We compute the normal subgroup $L=\gen{\{v^{\tau(k)}-v \, \vert \, v \in N_1, k \in K \}} \triangleleft N_1$, and determine its order in polynomial time.

If $|N_1/L|=|N_2|$ then there is an epimorphism from $G_1$ to $G_2$ if and only if the groups $N_1/L \rtimes_{\tau} H_1/K$ and $N_2 \rtimes_{\gamma} H_2$ are isomorphic\footnote{The quotient $H_1/K$ naturally inherits the action $\tau$ since $K$ acts trivially on $N_1/L$ (see, e.g., proof of \Cref{thm-Taunt-Surjection})} (see the proof of \Cref{thm-Taunt-Surjection}). There exists a polynomial-time algorithm to test whether these two groups are isomorphic and, if so, to compute an explicit isomorphism between them by Le Gall \cite[Theorem 1.2]{LeGall_classical_iso}. If such an isomorphism exists, we can lift it to an epimorphism from $G_1$ to $G_2$ by composing the canonical epimorphism, $G_1 \to N_1/L \rtimes_{\tau} H_1/K$, with the computed isomorphism from $N_1/L \rtimes_{\tau} H_1/K$ to $G_2$.

Otherwise, we may assume $|N_1/L|>|N_2|$ (for if $|N_2| > |N_1 / L|$ then the algorithm can simply return that there is no epimorphism). Since $L \triangleleft N_1$ and $N_1$ is an $H_1$-module, it follows that $L$ and $N_1/L$ both are an $H_1$-module. 
Let $\bigoplus V_i$ and $\bigoplus U_j$ denote the decomposition of $N_1/L$ and $N_2$ into their indecomposable $H_1$-submodules and $H_2$-submodules, respectively. These decompositions can be computed in polynomial time by \Cref{lem:compute-max-indecomposable}. Note that the underlying group of each $V_i$ is $(\Z/p_{i}^{k_i} \Z)^{d_i}$, for some $p_i$, and the exponent of $V_i$ divides $p_i^{k_i}$ (see, e.g., \cite{thevenaz1981,GrochowLevet}). Moreover, the underlying group of the quotient $V_i/p_iV_i$ is an elementary Abelian group $(\Z/p_i \Z)^{d_i}$ \cite{thevenaz1981}. 

Let $\bar{\phi}:H_1/K \to H_2$ be an isomorphism. All such isomorphisms can be computed in polynomial time, as both $H_1/K$ and $H_2$ are cyclic groups. Fix such an isomorphism $\bar{\phi}$. For each $i$, the module $V_i$ inherits an $H_2$-module structure via $\bar{\phi}$, since $K$ acts trivially on $V_i$. 

We now construct a bipartite graph $\Gamma_{\overline{\phi}}$ as follows. The vertex set of $\Gamma_{\overline{\phi}}$ is the disjoint union of two sets, one containing a vertex labeled $i$ for each $V_i$, and the other containing a vertex labeled $j$ for each $U_j$. There is an edge $(i,j)$ between vertices $i$ and $j$ if and only if there exists $l_{ij} \leq k_i$ such that there is an $H_2$-module isomorphism from $V_i/p_i^{l_{ij}}V_i$ to $U_j$.

\begin{claim}\label{claim:moduleiso}
For some $l_{ij} \leq k_i$ there is an $H_2$-module isomorphism from $V_i/p_i^{l_{ij}}V_i$ to $U_j$ if and only if there is an $H_2$-modules isomorphism from $V_i/p_iV_i$ to $U_j/p_iU_j$, and $\ep(U_j) \leq \ep(V_i)$.
\end{claim}

\noindent \emph{\bf Proof of \cref{claim:moduleiso}}.
Let $\rho_{ij}:V_i/p_i^{l_{ij}}V_i \to U_j$ be an $H_2$-module isomorphism. Note that  $(\Z/p_{i}^{k_i} \Z)^{d_i}$ and $(\Z/p_{i}^{k_i'} \Z)^{d_j}$ are the underlying groups of $V_i$ and $U_j$, respectively. Additionally, we have $(\Z/p_{i}^{l_{ij}} \Z)^{d_i} \cong (\Z/p_{i}^{k_i} \Z)^{d_i}/p_i^{l_{ij}}(\Z/p_{i}^{k_i} \Z)^{d_i} = V_i/p_i^{l_{ij}}V_i \cong U_j =(\Z/p_{i}^{k_i'} \Z)^{d_j}$ implies that $d_i=d_j$ and $k_i'=l_{ij}$. Define $f_{ij}: V_i/p_iV_i \to U_j/p_iU_j$ by $f_{ij}(\textbf{v}+p_iV_i)= \rho_{ij}(\textbf{v} \mod p_i^{l_{ij}}V_i) + p_iU_j$, where $\textbf{v}=(v_1,\ldots,v_{d_{i}}) \in V_i$ and $V_i/p_iV_i \cong (\Z/p_{i}\Z)^{d_i}$. It is easy to see that $f_{ij}$ is an epimorphism. Moreover, $f_{ij} ({\bf{v}}^h+p_iV_i)=\rho_{ij}({\bf{v}}^h \mod p_i^{l_{ij}}V_i) + p_iU_j={\rho_{ij}(\textbf{v} \mod p_i^{l_{ij}}V_i)}^h + p_iU_j=f_{ij} ({\bf{v}}+p_iV_i)^h$, where $h \in H_2$. Therefore, by Schur’s lemma \cite[Proposition 4]{serre1977} we have $f_{ij}$ is an $H_2$-module isomorphism.

Next, we have for each $v \in V_i/p_i^{l_{ij}}V_i$, there is $v' \in V_i$ such that $\ord(v')=\ord(v) m_v$ for some positive integer $m_v$. This implies $\text{exponent}(U_j)=\text{exponent}(V_i/p_i^{l_{ij}}V_i)={\rm lcm}(\ord(v) \,\vert\, v \in V_i/p_i^{l_{ij}}V_i)={\rm lcm}(\ord(v')/m_v \,\vert\, v' \in V_i) \leq {\rm lcm}(\ord(v') \,\vert\, v' \in V_i)=\text{exponent}(V_i)$. 

Conversely, let $f_{ij}:V_i/p_iV_i \to U_j/p_iU_j$ be an $H_2$-module isomorphism and let $\ep(U_j) \leq \ep(V_i)$. Let $(\Z/p_{i}^{k_i} \Z)^{d_i}$ and $(\Z/p_{i}^{k_i'} \Z)^{d_j}$ be the underlying groups of $V_i$ and $U_j$, respectively. Since $f_{ij}$ is an isomorphism, we have $d_i=d_j$. Take $l_{ij}=k_i' \leq k_i$. By \cite[Theorem 1.1]{thevenaz1981}, if modules $V_i/p_i^{l_{ij}}V_i$ and $U_j$ are not isomorphic then so the modules $(V_i/p_i^{l_{ij}}V_i) / (p_i V_i/p_i^{l_{ij}}V_i)$ and $U_j/p_i U_j$. This completes the proof of the claim.
\qed

First, we see that the graph $\Gamma_{\overline{\phi}}$ and the isomorphism $\rho_{ij}$ can be constructed in deterministic polynomial time. By \Cref{claim:moduleiso}, to test whether there is an $H_2$-module isomorphism from $V_i/p_i^{l_{ij}}V_i$ to $U_j$ it is enough to test whether there is an $H_2$-modules isomorphism from $V_i/p_iV_i$ to $U_j/p_iU_j$, and that $\ep(U_j) \leq \ep(V_i)$. Since the Cayley tables of the underlying groups of $V_j$ and $U_j$ can be constructed, we can compute $\ep(U_j)$ and $\ep(V_i)$ in polynomial time. 

Both $V_i/p_iV_i$ and $U_j/p_iU_j$ are elementary Abelian $p_i$-groups. Let $M_{V_i}, M_{U_j} \in {\GL}(\mathbb{F}_{p_i},{d_i})$ be the matrices induced by an action of $H_2$ on $V_i/p_iV_i$ and $U_j/p_iU_j$ respectively. Testing whether there is an $H_2$-modules isomorphism between these modules is equivalent to testing whether $M_{V_i}$ and $M_{U_j}$ are conjugate in ${\GL}(\mathbb{F}_{p_i},{d_i})$. Since  $\GL(\field_p,n)$ is an $\field_p$-algebra,
this conjugacy problem can be solved in deterministic polynomial time as described in \cite[Theorem 2]{CIK97}. Moreover, this algorithm also compute a matrix $T$ such that $T^{-1} M_{V_i} T=M_{U_j}$, if such a $T$ exists. Then $f_{ij}$ is the corresponding linear map. Moreover, we can compute a module isomorphism $\rho_{ij}$ by \cite[Corollary 1.2]{Teodorescu}. Therefore, the graph $\Gamma_{\overline{\phi}}$ and the isomorphism $\rho_{ij}$ can be constructed in deterministic polynomial time.

Next, we compute a maximum matching in $\Gamma_{\overline{\phi}}$, which can be done in polynomial time. If the matching saturates all the vertices corresponding to the $U_j$'s, i.e., if its size equals the number of indecomposable $H_2$-modules of $N_2$, then the matching defines the epimorphism from $N_1/L$ to $N_2$, say $\overline{\rho}$ (otherwise, we conclude that no epimorphism from $N_1/L$ to $N_2$ exists and hence there is also no epimorphism from $G_1$ to $G_2$). Moreover, since we have computed the homomorphism $\rho_{ij}$ corresponding to the matched edges, we can explicitly construct $\overline{\rho}$ in polynomial time. Note that $(\overline{\rho}, \bar{\phi})$ gives an isomorphism from $N_1/L \rtimes_{\tau} H_1/K$ to $G_2$ since $\overline{\rho}$ and $\bar{\phi}$ satisfies \Cref{equation} by construction. Next, we lift this isomorphism to an epimorphism from $G_1$ to $G_2$ by composing the canonical epimorphism $G_1 \to N_1/L \rtimes_{\tau} H_1/K$ with the computed isomorphism. \qed

\section{Efficient epimorphism testing of $\Hall{\prod \ElemAb}{\ElemAb}$}\label{sec-EE}

In this section, we solve $\GrS$ for groups from $\Hall{\prod \ElemAb}{\ElemAb}$. We first prove \Cref{thm-EEtoSInduced} in \Cref{sec:EEtoSInduced}. Then we further study $\RepEpi$ in \Cref{subsec-thm3} where we prove \Cref{thm-RepEpiTOCode}. We conclude this section by proving \Cref{thm-ProdElem}.

\subsection{Proof of \Cref{thm-EEtoSInduced}}\label{sec:EEtoSInduced}
\noindent \emph{Proof of \Cref{thm-EEtoSInduced}.} \emph{$\GrS$ of Groups in $\Hall{\ElemAb}{\ElemAb}$ to $\RepEpi$.} 
By listing all normal Hall subgroups and their complements, we can find a normal Hall subgroup $N_1$ of $G_1$ (resp. $N_2$ of $G_2$) with complement $H_1$ (resp. $H_2$), such that both $N_1$ and $N_2$ are elementary Abelian $p$-groups, both $H_1$ and $H_2$ are elementary Abelian $q$-groups. Thus, in order to test the epimorphism of the $G_1$ and $G_2$, we first solve the group epimorphism problem for the normal and complement parts. Since the normal and complement parts of $G_1$ and $G_2$ are both elementary abelian, their epimorphism problems can be solved in polynomial time by checking their orders. Given this, the only task left is to test whether there exist epimorphisms $\rho:N_1 \to N_2$ and $\phi:H_1 \to H_2$ satisfying $\Cref{equation}$, which can be reduces to solving $\RepEpi$ by considering the representations $\tau$ and $\gamma$ as an instance of the problem $\RepEpi$.

\emph{$\RepEpi$ to $\GrS$ of Groups from $\Hall{\ElemAb}{\ElemAb}$.} As discussed in \Cref{sec-prelims}, a representation $\tau:H \to \aut(N)$, defines a group $G=N \rtimes_{\tau} H$. Given two representations, $\tau: H_1 \to \aut(\Z_p^{m_1})$ and $\gamma: H_2 \to \aut(\Z_p^{m_2})$, we can construct two groups $G_1=\Z_p^{m_1} \rtimes_{\tau} H_1$ and $G_2=\Z_p^{m_2} \rtimes_{\gamma} H_2$. Since $H_i$'s are elementary Abelian $q$-groups, both $G_1,G_2 \in \Hall{\ElemAb}{\ElemAb}$. Then we can call an oracle to test if there is an epimorphism from $G_1$ to $G_2$. By \Cref{thm-Taunt-Surjection}, the two representations are $\Epi$ if and only if there is an epimorphism from $G_1$ to $G_2$, which gives the reduction.
\qed

\subsection{Proof of \Cref{thm-RepEpiTOCode}}\label{subsec-thm3}

\textbf{Representations of $\mathbb{Z}_q^l$ over $\mathbb{Z}_p$.} Before we proceed to prove \Cref{thm-RepEpiTOCode}, we recall some basic facts on representations of $\mathbb{Z}_q^{l}$ over $\mathbb{Z}_p$, where $p$ and $q$ are two different primes. The reader can refer to \cite[Section A5]{QST} for details.

Let $\Phi_q(x)$ be the $q^\text{th}$ cyclotomic polynomial. Let $h_1(x)\cdots h_r(x)$ are factors of $\Phi_q(x)$ over $\mathbb{Z}_p$ such that $h_i$'s are monic polynomial of same degree $m=(q-1)/r$, where $m$ is the order of $p$ in the multiplicative group $(\mathbb{Z}/q\mathbb{Z})^{\times}$. Let $M \in \mathrm{GL}(\mathbb{Z}_p,d)$ be the companion matrix of $h_1(x)$. For each nonzero vector $v \in \mathbb{Z}_q^l$, define $ v^\star : \mathbb{Z}_q^l \to \mathbb{Z}_q, v^\star(u) = (v,u) \text{ (inner product of $u$ and $v$)} $, and define a representation $f_v : \mathbb{Z}_q^l \to \mathrm{GL}(\mathbb{Z}_p,d), f_v(u) = M^{v^\star(u)}$. Then $f_v$ is an irreducible representation of $\mathbb{Z}_q^l$ over $\mathbb{Z}_p$, and $\{f_v \, \vert \, v\in \mathbb{Z}_q^l\}$ is the set of all irreducible representations of $\mathbb{Z}_q^l$ over $\mathbb{Z}_p$.

 

Let $\tau: \mathbb{Z}_q^l \to \mathrm{GL}(\mathbb{Z}_p,m_1)$ and $\gamma: \mathbb{Z}_q^l \to \mathrm{GL}(\mathbb{Z}_p,m_2)$ be two representations of $\mathbb{Z}_q^l$ and let $m_2 \leq m_1$. By Maschke’s theorem (see e.g., \cite[Theorem 1]{serre1977}, both $\tau$ and $\gamma$ can be written as direct sum of irreducible representations as follows, $\tau \cong f_{v_1}^{ k_1} \oplus \cdots \oplus f_{v_t}^{ k_t}, v_i \in \mathbb{Z}_q^l$ and $\gamma=f_{u_1}^{r_1}\oplus \cdots \oplus f_{u_{t^{'}}}^{r_{t^{'}}}, u_i \in \mathbb{Z}_q^l$.

\begin{lemma}\label{lem:irrep}
With the notation above, $\tau$ and $\gamma$ are $\Epi$ if and only if for every irreducible representation $f_{u_j}$ of $\gamma$ with multiplicity $r_{j}$, there is an irreducible representation $f_{v_i}$ of $\tau$ equivalent to $f_{u_j}$ with multiplicity at least $r_{j}$ .
\end{lemma}
\begin{proof}
Assume that $\tau$ and $\gamma$ are $\Epi$. Let $T:  \Z_p^{m_1} \to \Z_p^{m_2}$ be a surjective linear map. Since both $ \Z_p^{m_1}$ and $\Z_p^{m_2}$ are both $\Z_q^{l}$-modules, $T$ is also a surjective $\Z_q^{l}$-module homomorphism. By fundamental theorem of module homomorphism, we have $\Z_p^{m_1}/ \ker(T) \cong \Z_p^{m_2}$ which implies that $\Z_p^{m_1} \cong \ker(T) \oplus \Z_p^{m_2}$. Thus, each $f_{u_j}$ appear in the decomposition of $\tau$ with multiplicity at least $r_j$ for $1 \leq j \leq t^{'}$. 
Conversely, if all irreducible constituents of $\gamma$ appear in $\tau$ with the appropriate 
multiplicities, then we can define a surjective $\mathbb{Z}_q^l$-module homomorphism 
$T: \mathbb{Z}_p^{m_1} \to \mathbb{Z}_p^{m_2}$ by projecting onto the relevant direct summand. \end{proof}

Let $\tau: \mathbb{Z}_q^{l_1} \to \mathrm{GL}(\mathbb{Z}_p,m_1)$ and $\gamma:  \mathbb{Z}_q^{l_2} \to \mathrm{GL}(\mathbb{Z}_p,m_2)$ be two representations. Let $f_{u}$ be an irreducible representation of $\Z_{q}^{l_2}$ over $\Z_p$, appears in the decomposition of $\gamma$. Let $\phi: \Z_q^{l_1} \to \Z_q^{l_2}$ be a surjective linear transformation. Then the induced representation of $f_u$ by $\phi$ gives the representation $f_u \circ \phi$. Moreover, $(f_u \circ \phi)(v)=f_{\phi^{T}(u)}(v)$, for all $v \in \mathbb{Z}_q^{l_1}$ \cite[Section 5.1]{QST}. 

We decompose $\tau=f_{v_1}^{k_1}\oplus \cdots \oplus f_{v_t}^{k_t}$, for $v_i \in \Z_{q}^{l_1}, i\in [t]$, and $\gamma=f_{u_1}^{r_1}\oplus \cdots \oplus f_{u_{t'}}^{r_{t'}}$, for $u_i \in \Z_{q}^{l_2}, i\in [t']$ into their irreducible representations such that $k_1+\dots+k_t=m_1$ and $r_1+\dots+r_{t'}=m_2$. This decomposition can be computed in polynomial time, as shown in
\cite[Proposition 2]{QST}. Consider a set $\mathcal{L}_1=\{v_1,\dots,v_1,\dots, v_t, \dots, v_t\}$ of vectors in $\Z_q^{l_1}$, where each $v_i$ appears exactly $k_i$ times in the set $\mathcal{L}_1$. Similarly, let $\mathcal{L}_2=\{u_1,\dots,u_1,\dots, u_{t'}, \dots, u_{t'}\}$  of vectors in $\Z_q^{l_2}$, where each $u_j$ appears exactly $r_j$ times in the set $\mathcal{L}_2$. Using $\mathcal{L}_1$ and $\mathcal{L}_2$, we construct two matrices, $M_1=(v_1,\dots,v_1,\dots, {v_t, \dots, v_t})$ of order $l_1\times m_1$, and $M_2=(u_1,\dots,u_1,\dots, u_t', \dots, u'_{t'})$ of order $l_2\times m_2$.

To test if there is an epimorphism from $G_1$ to $G_2$, where $G_1$ and $G_2$ identified as $G_1=\Z_p^{m_1} \rtimes_{\tau} \Z_q^{l_1}$ and $G_2=\Z_p^{m_2} \rtimes_{\gamma} \Z_q^{l_2}$, by \Cref{thm-EEtoSInduced} we can view $\tau$ (resp. $\gamma$) as  representations of $\Z_{q}^{l_1}$ (resp. of $\Z_{q}^{l_2}$) over $\Z_p$ of dimension $m_1$ (resp. $m_2$). We solve $\RepEpi$ problem for $\tau$ and $\gamma$. \\

\noindent \emph{\bf Proof of \Cref{thm-RepEpiTOCode}.}  By \Cref{lem:irrep}, $\tau$ and $\gamma \circ \phi$ are $\Epi$ if and only if for each $j \in [t']$, there exists $i \in [t]$ such that $f_{\phi^{T}(u_{j})}$ is equivalent to $f_{v_{i}}$ and $r_{j} \leq k_{i}$. This latter condition is equivalent to verifying that $\{\phi^T(u_1),\dots,\phi^T(u_{r_{t'}})\} \subseteq \{v_1,\dots,v_t\}$ as each irreducible representation of $\Z_q^{l_1}$ is of the form $f_{v_i}$ for some $v_i \in \Z_q^{l_1}$. To test whether $\{\phi^T(u_1),\dots,\phi^T(u_{r_{t'}})\} \subseteq \{v_1,\dots,v_t\}$, it is enough to compute the linear transformation $\phi$, a matrix $P_{m_1\times m_2}$ such that each row and each column of $P$ contains at most one entry equals to $1$ (with all other entries being zero), and the equality $M_1 P= \phi^{T} M_2$ holds. We now prove that finding such $\phi$ and $P$ can be reduced to an instance of $\codeIso$.

Without loss of generality, we may assume that $\mathrm{rank}(M_1) = l_1$ (resp. $\mathrm{rank}(M_2) = l_2$).  We can perform row operations on $M_1$ (resp. on $M_2$) followed by
a permutation of the columns to transform $M_1$ (resp. $M_2$) to a
standard form. Let $A_1 \in \GL(\mathbb{Z}_q,l_1)$ (resp. $A_2 \in \GL( \mathbb{Z}_q,l_2)$) and $B_1$ a permutation matrix in $\GL(\mathbb{Z}_q,m_1)$ (similarly $B_2 \in \GL(\mathbb{Z}_q,m_2)$) be such that $A_1 M_1B_1=[I_{l_1}|M_1']$ (resp. $A_2 M_2B_2=[I_{l_2}|M_2']$). Since $P$ is an injection matrix, we can write $P=SQ$, where $S=\begin{bmatrix}
    I_{m_2}\\
    \bf{0}
\end{bmatrix}$ is of dimension ${m_1 \times m_2}$ and $Q$ is a permutation matrix of dimension $m_2\times m_2$.

Note that the number of zero columns of $M_2$ can be at most the number of zero columns of $M_1$, as otherwise, there can not exist $P$ and $\phi^T$ satisfying $M_1 P= \phi^{T} M_2$ (in which case the algorithm immediately rejects). Therefore, once we get to this point, we can safely delete all the zero columns of $M_1$ and $M_2$. Now 
\begin{align}
   M_1 P &= \phi^{T} M_2  \nonumber \\
   A_1^{-1}[I_{l_1} \mid M_1']B_1^{-1} SQ&=\phi^{T}A_2^{-1}[I_{l_2} \mid M_2']B_2^{-1} \nonumber\\
   [I_{l_1} \mid M_1']B_1^{-1}SQ B_2&=A_1 \phi^{T} A_2^{-1}[I_{l_2} \mid M_2'] \label{eq:last}
\end{align}

Note that $\rk([I_{l_1} \mid M_1']B_1^{-1}S)=\rk(M_1P)=\rk(\phi^TM_2)=l_2$. Let $\psi$ be the matrix such that $\psi [I_{l_1} \mid M_1']B_1^{-1}S=\begin{bmatrix} \tilde{M_1}\\ 
\mathbf{0}
\end{bmatrix}$, where $\tilde{M_1}$ is of dimension ${(l_2\times m_2)}$ and $\rk(\tilde{M_1})=l_2$. Let $\tilde{Q}=QB_2$. We have $ \psi  [I_{l_1} \mid M_1']B_1^{-1}SQ B_2=\begin{bmatrix} \tilde{M_1} \tilde{Q}\\ 
\mathbf{0}
\end{bmatrix}$. By \Cref{eq:last}, there exists $\tilde{\phi^T}$ such that $ \psi  A_1 \phi^{T} A_2^{-1}[I_{l_2} \mid M_2']=  \begin{bmatrix}
\bar{\phi^T}\\
\mathbf{0}
\end{bmatrix}  [I_{l_2} \mid M_2']$, where $\tilde{\phi^T}_{(l_2\times l_2)}$ is of rank $l_2$. It follows now from \Cref{eq:last} that 
\begin{align}
 \tilde{M_1} \tilde{Q} =\bar{\phi^T}  [I_{l_2} \mid M_2']  \label{eq:CodeIso}
\end{align}

\Cref{eq:CodeIso} is an instance of the $\codeIso$ problem, completing the proof. \qed

\Cref{cor-ElemElem} then follows from \Cref{thm-EEtoSInduced} and \Cref{thm-RepEpiTOCode}. At that point, Babai's algorithm for $\codeIso$ in singly exponential time \cite{GrIBCGQ} gives a polynomial-time algorithm for $\RepEpi$ for elementary Abelian groups, finishing the proof of \Cref{cor-ElemElem}.

Next, we prove \Cref{thm-ProdElem}.\\

\noindent \emph{\bf Proof of \Cref{thm-ProdElem}.} The idea for $\Hall{\ElemAb}{\ElemAb}$ can be extended to $\Hall{\prod \ElemAb}{\ElemAb}$ as follows. 
First we check that $G_2$ is in $\Hall{\prod \ElemAb}{\ElemAb}$ (otherwise reject). For both $G_1$ and $G_2$, we compute all the normal Hall subgroups and their complements in polynomial time by \cite[Theorem 1]{QST}. Additionally, we can identify $G_1, G_2$ as $G_1=\prod_{i\in S_1}\Z_{p_i}^{k_{i}} \rtimes_{\tau} \Z_q^{l_1}$ and $G_2=\prod_{i\in S_2}\Z_{p_i}^{k_{i}'} \rtimes_{\gamma} \Z_q^{l_2}$ with the associated actions as $\tau$ and $\gamma$, respectively. 

Now we need to test if there exist an epimorphism $\psi: \prod_{i\in S_1}\Z_{p_i}^{k_{i}} \rightarrow \prod_{i\in S_2}\Z_{p_i}^{k_{i}'} $ and $\phi: \Z_q^{l_1} \rightarrow \Z_q^{l_2} $ such that $\psi \circ \tau(h_1)= \gamma(\phi(h_1)) \circ \psi, \forall h_1 \in \Z_q^{l_1}$. Let  $\tau_i: \Z_q^{l_1}\rightarrow \GL(\Z_{p_i},k_{i})$ be the projection of $\tau$ into the $i$-th component for $i\in S_1$, and let $\gamma_i: \Z_q^{l_2}\rightarrow \GL(\Z_{p_i},k_{i}')$ be the projection of $\gamma$ into the $i$-th component for $i\in S_2$.
For each $\tau_i$ and $\gamma_i$, let $M_i$ and $M_{i'}$ be the matrices arising from the irreducible representation $\tau_i$ and $\gamma_i$, respectively (see proof of \Cref{thm-RepEpiTOCode}). Define $M_1:=\operatorname{diag} (M_{i})_{i\in S_1}, M_2=\operatorname{diag} (M_{i'})_{i'\in S_2} $, i.e., $M_1$ (resp. $M_2$) is a block diagonal matrix with diagonal blocks $M_i$ (resp. $M_{i'}$). Then it suffices to solve $M_1P=\phi M_2$ as in \Cref{thm-RepEpiTOCode}. \qed

\section{Efficient epimorphism testing in groups with limited Abelian composition factors}\label{sec-wtAbCF}
In this section, we prove \Cref{thm-wtAbCF}. We first prove the following lemma. 


\begin{lemma} \label{lem:nonabelian}
Suppose $H$ is a quotient of $G$, and let $\ell$ be the socle length of $G$. Then there exist integers $k_1,\dotsc,k_\ell$ and groups $M_{ij}$ for $i=1,\dotsc,\ell$ and $j=1,\dotsc,k_i$ for all $i$ such that
\begin{enumerate}
\item $M_{ij}$ is a minimal normal subgroup of $G / \Soc^{i-1}(G)$
\item There is a minimal normal subgroup $\widehat{M}_{ij}$ of $G_{i} := G / \langle \widehat{M}_{i'j'} : i' < i \text{ and } j' \in [k_{i'}]\rangle$ such that the natural quotient map $G_i \to G / \Soc^{i-1}(G)$ maps $\widehat{M}_{ij}$ isomorphically onto $M_{ij}$.
\item $H \cong G_{\ell+1}$.
\end{enumerate}
\end{lemma}

\begin{proof}
Let $K \unlhd G$ be such that $G/K \cong H$. Let $M_{11},\dotsc,M_{1,k_1}$ be the minimal normal subgroups of $G$ that are contained in $K$. Then we can define $\widehat{M}_{1j} = M_{1j}$ to satisfy the first two parts of the lemma. Define $K_1 = \langle M_{1j} | j \in [k_1] \rangle$.

Now let $i > 1$, and suppose inductively that $M_{i'j'}$, $\widehat{M}_{i'j'}$ have been defined for all $i' < i$ and all $j' \in [k_{i'}]$ that satisfy (1) and (2) and such that $M_{i',1},\dotsc,M_{i',k_{i'}}$ are all the minimal normal subgroups of $G / \Soc^{i'-1}(G)$ that are contained in $K / \Soc^{i'-1}(G)$. Define $K_{i-1} = \langle M_{i'j'} | i' <i; j' \in [k_{i'}] \rangle$, and let $\widehat{M}_{i1},\dotsc,\widehat{M}_{i,k_i}$ be the minimal normal subgroups of $G/K_{i-1}$ that are contained in $K / K_{i-1}$. First, we claim that they are mapped isomorphically onto their images by the natural quotient map $G / K_{i-1} \to G / \Soc^{i-1}(G)$; call these images $M_{ij}$. For if some $\widehat{M}_{ij}$ is not mapped isomorphically onto its image by this map, then it must be contained in the kernel of the  map, hence in $\Soc^{i-1}(G)$. But since $\widehat{M}_{ij}$ was contained in $K / K_{i-1}$, if it were also in $\Soc^{i-1}(G)$, then by construction $\widehat{M}_{ij}$ would have been contained in $K_{i-1}$ (since the latter is equal to $K_i \cap \Soc^{i-1}(G)$), contradicting the fact that it was a minimal normal subgroup of $G / K_{i-1}$.
Let $M_{ij}$ be the image of $\widehat{M}_{ij}$ in $G / \Soc^{i-1}(G)$. Then $M_{ij}$ is indeed a minimal normal subgroup of $G / \Soc^{i-1}(G)$. 

The above induction continues until $K / K_i = 1$, at which point $K = K_i$. Since each time we progress up the socle series, this happens at some $i \leq \ell+1$, and thus we satisfy part (3) of the lemma.
\end{proof}


\noindent \emph{\bf Proof of \Cref{thm-wtAbCF}.} 
First, we check that $G_2$ has no Abelian composition factors (otherwise reject). Computing the composition factors in $\poly(|G_2|)$ time is easy (indeed, computing the composition factors of permutation groups can even be done in $\mathsf{NC}$ \cite{CSinNC}).

Next, we compute a chief series of $G_1$ that refines the socle series, which is again easy by standard techniques. 

For all subsets of the chief factors in our computed chief series---note that the chief series has length at most $\log |G_1|$, hence that are at most $2^{\log |G_1|} \leq |G_1|$ many such subsets---we attempt to make those the kernel of an epimorphism $G_1 \to G_2$. Say the subset we pick has $M_{ij}$ for $j=1,\dotsc,k_i$ that are among the minimal normal subgroups of $G_1 / \Soc^{i-1}(G_1)$.

First, we check that the order is correct: if $|G_1| / \prod_{i,j} |M_{ij}| \neq |G_2|$, then skip this choice of subset of the chief factors.

Next, we check that this subset of chief factors is a valid kernel, in the sense of Lemma~\ref{lem:nonabelian}. That is, we check inductively that there are subgroups $\widehat{M}_{ij}$ corresponding to the chosen $M_{ij}$ that satisfy parts (1) and (2) of Lemma~\ref{lem:nonabelian}. If so, then we finally we compute the quotient $G_1 / \prod_{i,j} \widehat{M}_{ij}$, and test isomorphism between that group and $G_2$. The latter can be done in polynomial time since $G_2$ is Fitting-free \cite{GrIBCQ}.
\qed

We can make some progress towards extending \Cref{thm-wtAbCF} to Fitting-free groups, and in doing so, we highlight one of the remaining obstacles. The issue is that when $\Soc(G)$ contains Abelian minimal normal subgroups $T_1,\dotsc,T_k$ such that the subgroup they generate is their direct product, and such that all $T_i$ are pairwise isomorphic as $G$-modules, then $T_1 T_2 \dotsb T_k \cong T_1 \times \dotsb \times T_k$ contains $p^{\Theta(k^2)}$ 
many Abelian normal subgroups $N$ (normal in $G$) that are isomorphic to $T_1^{k/2}$ as $G$-modules. (In the ``extreme'' case that the $T_i$ are central, they are each 1-dimensional, trivial $G$-modules $\Z_p$, and the previous discussion is the same as saying that a $k$-dimensional vector space has a number of $k/2$-dimensional subspaces that is $p^{\Theta(k^2)}$.) Of course we have $k \leq \log_p |G|$, but if $k = \Theta(\log |G|)$, then $p^{\Theta(k^2)} = |G|^{O(\log |G|)}$, so there are too many choices to brute force over them all in polynomial time. And then this issue can arise at any layer of the socle series.

Towards this direction, we make the following definition. Given an Abelian minimal normal subgroup $T \unlhd G$, we define its \emph{multiplicity} in $G$ as the largest $k$ such that there exist $T_1 \times \dotsb \times T_k \unlhd G$ where all $T_i$ are minimal normal subgroups of $G$ such that they are all isomorphic to $T$ as $G$-modules. 

\begin{definition}
The \emph{Abelian width} of a group $G$ is the maximum $k$ such that there is an $i \geq 0$ and an Abelian minimal normal subgroup $T \unlhd G / \Soc^i(G)$ of multiplicity $k$.
\end{definition}

We note that the Abelian width is always at most $O(\log|G|)$, and there are even Fitting-free groups that achieve this width, for example: $G=S_5^k$ has socle $A_5^k$, with quotient $\Z_2^k$, so its Abelian width is $k = \log_{120}|G|$. 

However, we also note that the Abelian width can be quite a bit smaller than the dimension of any chief factor, as in the following example. For sufficiently large $p$, prime cyclic group $C_p$ has $p$ distinct irreducible representations over $\F_{3^{p-1}}$ (an extension of $\F_3$ that has $p$-th roots of unity), that are 1-dimensional over $\F_{3^{p-1}}$ (hence dimension $p-1$ over $\F_3$). If we let $N$ be the direct sum of those irreps, then the group $G = N \rtimes C_p$ has order $3^{p(p-1)} p$, the socle is $N$ which has dimension $\Theta(p^2)$ over $\F_3$, each irrep has dimension $p-1 = \Theta(\sqrt{\log|G|})$ over $\F_3$, but the Abelian width is only 1 since $N$ is a direct sum of pairwise non-isomorphic, irreducible $C_p$-modules (hence, the same is true when viewed as $G$-modules).

The following is a corollary to the proof of \Cref{thm-wtAbCF} above.

\begin{corollary} \label{cor:Abwidth}
There is a polynomial-time algorithm for $\GrS$, restricted to pairs of groups $(G_1,G_2)$ such that $G_2$ is Fitting-free and $G_1$ has Abelian width at most $O(\sqrt{\log |G_1|})$.
\end{corollary}

\begin{proof}[Proof sketch]
The idea is essentially the same as above, except we handle Abelian subgroups in each socle layer a little differently than the chief factors as in the proof of \Cref{thm-wtAbCF}. If $\Soc(G / \Soc^{i-1}(G))$ contains an Abelian $G / \Soc^{i-1}(G)$-module $T$ of multiplicity $k$, then instead of breaking $T^k$ up into its chief factors (isomorphic copies of $T$), we brute force over all possible copies of $T^j$ contained within $T^k$, for $j=0,1,\dotsc,k$, of which there are at most $kp^{\Theta(k^2)} \leq |G|^{O(1)}$ by our assumption that the Abelian width is $O(\sqrt{\log |G_1|})$. The rest of the algorithm and proof are then the same as that of \Cref{thm-wtAbCF}.
\end{proof}

\section{Conclusion and open problems} \label{sec:conclusion}
Of course the main question we leave open is whether \GrS is in $\mathsf{P}$ (a positive answer would imply \GpI is in $\mathsf{P}$). Other basic questions about its complexity are:

\begin{open}
Is \GrS in $\mathsf{coAM}$?
\end{open}

\begin{open}
What is the complexity of the counting problem $\# \GrS$?
\end{open}

As stepping stones towards understanding the complexity of \GrS in general, other classes of groups for which \GpI is in $\mathsf{P}$ form natural targets for putting \GrS into $\mathsf{P}$:

\begin{open}
Is \GrS in $\mathsf{P}$ when the input groups are:
\begin{enumerate}
\item Fitting-free (cf. \cite{GrIBCQ, GrIBCGQ})?

\item When the target group has size $O(1)$? Or is $O(1)$-generated?

\item Quotients of genus 1 (cf. \cite{GrIWilson})? Both are of genus 2 (cf. \cite{GrIGenus2})?

\item Groups with Abelian Sylow towers (cf. \cite{BabaiQiao}) or tame towers (cf. \cite{GQtame})?

\item Groups with central, elementary Abelian radicals such that $G/Rad(G)$ is a direct product of non-Abelian simple groups (cf. \cite[Thm.~C]{GQcoho})?
\end{enumerate}
\end{open}

We comment on the Fitting-free question. If the target group $H$ is Fitting-free, then there is an epimorphism $G \to H$ iff there is an epimorphism $G / Rad(G) \to H$ (since if any part of the radical is not in the kernel, then it would yield an Abelian normal subgroup in the image), so we immediately reduce to the case that both groups are Fitting-free. If $G,H$ are Fitting-free, deciding whether $H$ is a quotient of $G$ by a subgroup of $\Soc(G)$ doesn't seem particularly hard; the difficulty arises when the socle of $H$ may come from chief factors that are at many different layers of the socle series of $G$.

\section*{Acknowledgments}
We thank the ICALP reviewers for their helpful comments on the submission. DT also thanks Armin Weiss for bringing the Group Epimorphism problem to her attention. Joshua A. Grochow was supported by NSF CAREER Award CCF-2047756. Dhara Thakkar was supported by JSPS KAKENHI Grant Numbers 24H00071 and 25K24674.
\bibliographystyle{alphaurl}
\bibliography{References}

\appendix

\section{Deferred proofs} \label{app}

Below we give a proof of \Cref{lem:max-indecomposable}. We recall the notation: $N_1$ is an Abelian group that is coprime to the cyclic group $H_1$, and $\tau \colon H_1 \to \aut(N_1)$ is an action. $K \unlhd H_1$, and $L = \{v^{\tau(k)} - v : v \in N_1, k \in K\}$. The observation to prove is that $L$ is a direct sum of maximal indecomposable $H_1$-submodules of $N_1$.

\noindent \emph{\bf Proof of \Cref{lem:max-indecomposable}.} 
Let $N_1= \bigoplus_{i} V_i$ be a decomposition of $N_1$ into its indecomposable $H_1$-submodules, $V_i$. Let $p_i$ be the prime such that $V_i$ is a $p_i$-group. 
Define $V_i' := \{v^{\tau(k)} - v : v \in V_i, k \in K\}$. Then we immediately have $L = \bigoplus_i V_i'$.
Then it is necessary and sufficient to show that $V_{i}^{'}=V_{i}$ whenever $V_{i}^{'} \neq 0$.

From \cite[Corollary 1.2]{thevenaz1981} it follows that, since $V_i$ is indecomposable and $|H_1|$ is coprime to $|N_1|$, the only submodules of $V_i$ are of the form $p_i^l V_i$ for some $l \geq 0$. 

Let $p_i^b$ be the exponent of $V_i$: $p_i^bV_i=0$ and $p_i^{b-1}V_i \neq 0$. Suppose that $V_i^{'}=p_i^l V_i$ for some $b > l \geq 1$ (in particular, we are assuming $V_i' \neq 0$, as that is the only case we need to handle). Then $K$ acts trivially on $V_i \mod p_i^l$. For $k \in K$, we have $\tau(k)=\iota + p_i^l M_k$, where $\iota$ is an identity map and $M_k \in \text{End}(V_i)$.  We now prove that $\tau(k)^{p_i^{b-l}}=\iota$. We have $\tau(k)^{p_i^{b-l}}=(\iota + p_i^l M_k)^{p_i^{b-l}}= \sum_{j=0}^{p_i^{b-l}} \binom{p_i^{b-l}}{j} (p_i^l M_k)^{j}$. To show that $\tau(k)^{p_i^{b-l}}=\iota$, it is enough to show that $p_i^b$ divides $\binom{p_i^{b-l}}{j} (p_i^l M_k)^{j}$ for all $j$. If $j \leq p_i$ then $p_i^b \vert \binom{p_i^{b-l}}{j}$. And if $j \geq b/l$ then $p_i^b \vert (p_i^l M_k)^j$. So now consider $p_i \leq j < b/l$. The power of $p_i$ dividing $j!$ is $\leq \lfloor j/p_i \rfloor+ \lfloor j/p_i^2 \rfloor+ \lfloor j/p_i^3 \rfloor + \cdots \leq \frac{j}{p_i-1}$. Thus, $p_i^{b-l-\frac{j}{p_i-1}} \vert \binom{p_i^{b-l}}{j}$. 

Now suppose $p_i \neq 2$ or $l \neq 1$. Then, we have $l(p_i-1) \geq \frac{j}{j-1}$, and therefore 
\begin{align*}
b - l - \frac{j}{p_i-1} + jl & = b + l(j-1) - \frac{j}{p_i-1} \\
 & \geq b.
 \end{align*}
It follows that, $p_i^{b} \vert \binom{p_i^{b-l}}{j} (p_i^l M_k)^{j}$. 

Finally, let us now consider the case when $p_i=2$ and $l=1$. It is standard that if $2^{x} \vert j!$, then $x \leq j-1$ \cite{StackExchange}. Therefore, $2^{b}=2^{(b-1)-(j-1)+j} \vert   \binom{2^{b-1}}{j} (2 M_k)^{j}$. Thus proves that  $\tau(k)^{p_i^{b-l}}=\iota$. This implies that $\ord(\tau(k)) \vert p_i^{b-l}$. However, since $\tau: H_1 \to \aut(N_1)$ is a homomorphism, $\ord(\tau(k)) \vert \ord(k)$, and $\ord(k)\vert |H_1|$. Since $\gcd(|N_1|, |H_1|)=1$, we get $\ord(\tau(k)) = 1$, hence $\tau(k)=\iota$. This implies that $V_i^{'}=0$, which is a contradiction. Therefore, $V_{i}^{'}=V_{i}$. This completes the proof of \Cref{lem:max-indecomposable}.
\qed

\end{document}